\documentclass[11pt]{article}

\usepackage{silence}
\usepackage[notimes]{arxiv}
\usepackage{amsmath,amssymb}
\usepackage{amsthm}
\usepackage{booktabs}
\usepackage{graphicx}
\usepackage[numbers]{natbib}
\usepackage{float}
\usepackage{tabularx}
\usepackage{hyperref}

\makeatletter
\renewcommand{\@maketitle}{%
  \vbox{%
    \hsize\textwidth
    \linewidth\hsize
    \vskip 0.1in
    \@toptitlebar
    \centering
    {\LARGE\bfseries \@title\par}%
    \@bottomtitlebar
    {\normalfont\undertitle}\par
    \vskip 0.1in
    \def\And{%
      \end{tabular}\hfil\linebreak[0]\hfil%
      \begin{tabular}[t]{c}\bf\rule{\z@}{24\p@}\ignorespaces%
    }%
    \def\AND{%
      \end{tabular}\hfil\linebreak[4]\hfil%
      \begin{tabular}[t]{c}\bf\rule{\z@}{24\p@}\ignorespaces%
    }%
    \begin{tabular}[t]{c}\bf\rule{\z@}{24\p@}\@author\end{tabular}%
  \vskip 0.4in \@minus 0.1in \center{\@date}   \vskip 0.2in
  }%
}
\renewenvironment{abstract}{%
  \centerline{\large\bfseries\arxivabstractname}%
  \begin{quote}%
}{%
  \end{quote}%
}
\makeatother
\AtBeginDocument{%
  \fancyhead[R]{\footnotesize\normalfont\headeright}%
  \raggedbottom
}

\hypersetup{
  colorlinks=true,
  linkcolor=blue,
  urlcolor=blue,
  citecolor=blue
}

\renewcommand{\arxivabstractname}{Abstract}
\renewcommand{\keywordname}{\textbf{Keywords}}

\title{SpanKey: Dynamic Key-Space Conditioning for Neural Network Access Control}
\author{WenBin Yan\\University of Colorado Boulder\\\texttt{wenbin.yan@colorado.edu}}
\date{\today}

\renewcommand{\shorttitle}{SpanKey}
\renewcommand{\undertitle}{}
\renewcommand{\headeright}{}
\theoremstyle{plain}
\newtheorem{proposition}{Proposition}

\begin{document}
\maketitle

\begin{abstract}
This paper introduces \emph{SpanKey}, a lightweight mechanism for neural network access control that conditions intermediate activations on a secret key---taking a different path than weight encryption or homomorphic inference. The core construction is simple. A deployer holds a secret basis matrix $B$ whose span defines a low-dimensional key subspace. At each inference call, coefficients $\alpha$ are sampled and a dynamic key $k = \alpha^\top B$ is injected into intermediate layer activations via additive or multiplicative mappings. Authorized users receive in-span keys; unauthorized keys are drawn from a high-dimensional distribution outside the subspace.

We build the argument across three interlocking lines. At the \emph{mechanism} level, we lay out the subspace-key injection design space---additive and multiplicative mappings, multi-layer placement, per-layer independent bases and coefficients---and introduce a family of deny losses: Mode~A (entropy maximization), Mode~B (an explicit reject class), and Mode~C (a margin hinge), with systematic comparisons on CIFAR-10 ResNet-18. At the \emph{diagnostic} level, we formalize \emph{key absorption}---the failure mode where correct-key-only supervision drives the network to suppress sensitivity to the injection direction in downstream layers. Propositions~\ref{prop:margin_tail} and~\ref{prop:beta_energy} frame this phenomenon through first-order margin analysis and subspace geometry. At the \emph{empirical} level, we aggregate nine experiments spanning MLPs, CNNs, ResNets, ViT-Tiny, GPT-2, and Qwen2.5, covering both classification and language modeling. Mode~B consistently achieves the most favorable trade-off between authorized accuracy and unauthorized degradation; Proposition~\ref{prop:reject_margin} explains why---the reject-dimension logit accumulates a $\log T$ gradient advantage that semantic dimensions cannot match. The security analysis makes no cryptographic claims. It maps out SpanKey's applicable boundary through classical probes and realistic white-box attacks (no-injection deployment, fine-tuning removal, representation-layer bypass): effectiveness hinges on the secrecy of $B$ and does not withstand weight leakage combined with a small amount of labeled data. All experiment code is publicly available at \url{https://github.com/mindmemory-ai/dksc}.
\end{abstract}


\section{Introduction}
\label{sec:intro}

Neural networks are widely deployed as online services, embedded systems, or APIs. The core requirement for \emph{inference gating} and \emph{access control} is straightforward: authorized users who hold a valid key should obtain usable predictions; unauthorized users should receive degraded, low-utility, or detectably corrupted output. Prior work approaches related problems through cryptographic protocols~\citep{homomorphic} or hardware isolation~\citep{tee} to protect model confidentiality; a separate line studies watermarking and intellectual-property claims~\citep{watermarking}. These routes differ from what we investigate here---a lightweight, model-native conditioning gate---in problem setting, security assumptions, and deployment cost.

SpanKey sits at the low-computation-overhead end of the design space. It does not encrypt weights, nor does it offer cryptographically meaningful confidentiality guarantees. The key is modeled as a \emph{conditioning variable} on the forward pass. The question we ask is: \emph{does there exist a supervised learning procedure that ties ``authorized behavior,'' in a statistical sense, to key injections from a secret subspace, while substantially degrading the behavior produced by keys sampled outside that subspace?}

\paragraph{How this differs from related work.}
Four neighboring research threads are worth distinguishing, each with a different problem setup or technical mechanism.

Homomorphic evaluation, secure multi-party computation, and TEE-based isolation~\citep{homomorphic,secureml,tee} target computational confidentiality in encrypted or isolated environments. The costs---latency, protocol complexity, key management infrastructure---are nothing like what we aim for: a lightweight injection-based gate operating on plaintext inference APIs. We explicitly do not claim cryptographic security.

Watermarking and model IP protection work~\citep{watermarking,encryip,modellock,deepip,textwatermarksurvey} centers on provenance, piracy detection, or copy control---post-hoc attribution, not runtime gating. SpanKey concerns the \emph{here and now}: whether the current inference carries a legitimate key determines whether the output is usable.

In the narrower direction of key-conditioned gating, K-OTG~\citep{kotg} applies orthonormal scrambling to LLM hidden states and trains a separate unauthorized-block path; EncryIP~\citep{encryip} embeds public-key logic into a label-randomization pipeline; ModelLock~\citep{modellock} perturbs the training distribution through diffusion-based editing. Compared to these, our construction is more direct: an explicit low-dimensional subspace $\mathrm{Span}(B)$, additive or multiplicative mid-level injection, no adapter modules or external encryption pipelines, and a natural connection to first-order margin analysis and deny-loss objectives.

Finally, semantic conditioning methods like FiLM~\citep{fiLM,gatedconv} modulate representations with observable task metadata. SpanKey differs in that the conditioning variable is \emph{secret} and dynamically sampled from a subspace---the optimization target is not multi-task accuracy but behavioral divergence between authorized and unauthorized forward passes.

Table~\ref{tab:related} compresses these directions into a side-by-side keyword comparison.

\begin{table}[t]
\centering
\footnotesize
\begingroup
\setlength{\tabcolsep}{2.5pt}
\hbadness=10000
\begin{tabularx}{\linewidth}{@{}>{\raggedright\arraybackslash}p{2.1cm}>{\raggedright\arraybackslash}X>{\raggedright\arraybackslash}X>{\raggedright\arraybackslash}X>{\raggedright\arraybackslash}X@{}}
\toprule
Method & Primary goal & Core mechanism & Deployment / integration & Contrast with this work \\
\midrule
SpanKey &
Inference-time access control &
Subspace key + mid-level injection; optional deny losses A--C / extensions &
Tunable $\gamma$, $m$, layer placement &
Absorption + Prop.~\ref{prop:margin_tail} diagnosis; deny-loss phenomenology \\

K-OTG \citep{kotg} &
LLM access control &
Orthonormal scrambling / inverse transform; unauthorized-block training &
Pre-output-head hooks; dual-path &
Hidden-state scrambling; adapter context \\

EncryIP \citep{encryip} &
IP / label protection &
Public-key label-space randomization &
Encrypt-decrypt label pipeline &
Cryptographic authorization labels, not representation injection \\

ModelLock \citep{modellock} &
Model locking &
Diffusion-based editing of training distribution &
Keyed editing for training + inference &
Distribution-pipeline dependency, not subspace injection \\

Deep IP \citep{deepip}; text watermarking \citep{textwatermarksurvey} &
Ownership / attribution &
Watermarks and fingerprints &
Varies by method &
Background surveys; not runtime gating \\

Secure inference \citep{homomorphic,secureml,tee} &
Confidential inference &
HE / SMC / TEE isolation &
Protocol stacks and system integration &
Target is \textbf{confidentiality}, not plaintext gating utility \\
\bottomrule
\end{tabularx}
\endgroup
\caption{\textbf{Qualitative} keyword-level comparison of SpanKey with representative directions; not a uniform empirical benchmark.}
\label{tab:related}
\end{table}

\paragraph{Paper organization.}
Section~\ref{sec:method} establishes notation, key generation, and the injection design space. Section~\ref{sec:theory} provides analytical motivation through first-order sensitivity and subspace geometry, with a dedicated analysis of Mode~B reject-margin convergence in \S\ref{sec:reject_margin}. All experiments live in Section~\ref{sec:experiments}: baseline separation, key absorption diagnosis, the deny-loss family (Modes A--C and extensions), Mode~B ablation, and cross-architecture validation (ViT-Tiny, Qwen2.5) alongside a generative task (GPT-2 language modeling). Section~\ref{sec:security} discusses the threat model and attack boundaries. Section~\ref{sec:conclusion} concludes.

\section{Method}
\label{sec:method}

\subsection{Notation and key generation}
Let $x \in \mathbb{R}^d$ denote the input, $y$ the label, and $f_\theta$ a network with parameters $\theta$. Let $B \in \mathbb{R}^{m \times d}$ be a basis matrix with $m \ll d$. The \emph{key subspace} is $\mathrm{Span}(B)$. At each training step, coefficients $\alpha \in \mathbb{R}^m$ are sampled and a \emph{dynamic key} is formed:
\[
k = \alpha^\top B .
\]
At inference time, a \emph{legitimate} key is any $k_{\text{infer}} \in \mathrm{Span}(B)$; an \emph{illegitimate} (wrong) key is sampled outside $\mathrm{Span}(B)$, e.g., from an isotropic Gaussian.

\subsection{Key injection}
The key is injected into one or more intermediate tensors. Let $h_\ell$ be the activation at layer $\ell$ (flattened to $\mathbb{R}^{d_\ell}$). A layer-compatible $k_\ell \in \mathbb{R}^{d_\ell}$ is constructed and applied as
\[
h_\ell' = \mathrm{Inject}(h_\ell, k_\ell).
\]
We study two basic injections:
\[
\mathrm{Inject}_{\text{add}}(h, k) = h + \gamma k,
\qquad
\mathrm{Inject}_{\text{mul}}(h, k) = h \odot (1 + \gamma \tanh(k)),
\]
where $\gamma$ controls the injection strength and $\tanh$ bounds the multiplicative scaling range. The construction is formally similar to gated convolutions~\citep{gatedconv} and FiLM~\citep{fiLM}, except the conditioning variable is a secret key rather than a semantic attribute.

Multi-layer injection introduces two orthogonal design choices: (i)~\emph{per-layer independent coefficients} (separate $\alpha_\ell$ per layer); (ii)~\emph{per-layer independent bases} (separate $B_\ell$ per layer). Both reduce cross-layer correlation of the injection signal and empirically help mitigate \emph{key absorption} (see Section~\ref{sec:theory}).

\subsection{Training objective structure}
The SpanKey training objective has two terms:
\begin{equation}
\mathcal{L} = \mathcal{L}_{\mathrm{CE}}^{\mathrm{ok}} + \lambda\,\mathcal{L}_{\mathrm{deny}},
\label{eq:total_loss}
\end{equation}
where $\mathcal{L}_{\mathrm{CE}}^{\mathrm{ok}}$ is the standard supervised loss on the correct-key path only (cross-entropy for classification, negative log-likelihood for language modeling), $\mathcal{L}_{\mathrm{deny}}$ is an optional deny-loss term (detailed in Section~\ref{sec:deny_family}), with weight $\lambda \ge 0$. Setting $\lambda = 0$ recovers the pure correct-key-supervised baseline.

\section{Theoretical motivation}
\label{sec:theory}

This section explains, in simplified terms, why the mechanism can work at all: how injection perturbations propagate to logits, why a low-dimensional subspace constitutes an operable control degree of freedom, and why baseline training may fail to produce strong key separation. The analysis below does \textbf{not} constitute a security theorem; it provides conditions under which separation is mathematically non-trivial.

\subsection{First-order sensitivity: how injection reaches the logits}
Assume a linear readout at the classifier head: $z = W h_L + b$, where $h_L$ is the penultimate-layer representation. The key is injected at an earlier layer $\ell$, giving $h_\ell' = h_\ell + \gamma k$ (using the additive form for clarity). Let $\psi$ denote the composite map from layer $\ell$ to $L$. A first-order Taylor expansion yields
\[
z(h_\ell + \gamma k) \approx z(h_\ell) + \gamma\, W J_\ell\, k,
\]
where $J_\ell = \partial h_L / \partial h_\ell'$ is the Jacobian. Define the effective Jacobian $W_{\mathrm{eff}} := W J_\ell$. The logit perturbation is $\Delta z \approx \gamma\, W_{\mathrm{eff}}\, k$. For a true class $y$ and a competing class $c$, the logit margin $M_{y,c} = z_y - z_c$ changes by
\[
\Delta M_{y,c} \approx \gamma\, (w_y - w_c)^\top J_\ell\, k.
\]
Multiplicative injection can be subsumed under the same framework by defining an effective increment $\delta = \gamma(h_\ell \odot k)$ at the injection point.

Two directly tunable levers emerge: (1) the scale of $\gamma k$ relative to typical activations---if $\|\gamma J_\ell k\|$ is tiny, $\Delta M$ is negligible; (2) the magnitude of $(w_y - w_c)^\top J_\ell k$---i.e., whether downstream computation amplifies the injection direction.

\subsection{Subspace geometry: a low-dimensional control channel}
The key subspace is low-dimensional by design ($m \ll d$). Legitimate keys $k = \alpha^\top B$ are confined to an $m$-dimensional linear subspace of $\mathbb{R}^d$, exercising only $m$ degrees of freedom. Wrong keys drawn from a broad distribution on $\mathbb{R}^d$ (e.g., isotropic Gaussian) carry the overwhelming majority of their energy in the orthogonal complement of $\mathrm{Span}(B)$.

\begin{proposition}[Wrong-key energy partition, Beta law]
\label{prop:beta_energy}
Let $P$ be the orthogonal projector onto an $m$-dimensional subspace of $\mathbb{R}^d$, and let $k \sim \mathcal{N}(0, I_d)$. Define the fraction of energy outside the subspace:
\[
\eta := \frac{\|(I-P)k\|_2^2}{\|k\|_2^2}.
\]
Then $\eta \sim \mathrm{Beta}\!\left(\frac{d-m}{2},\,\frac{m}{2}\right)$, with expectation $\mathbb{E}[\eta] = \frac{d-m}{d}$ and variance $\mathrm{Var}(\eta) = \frac{2m(d-m)}{d^2(d+2)}$.
\end{proposition}

Proposition~\ref{prop:beta_energy} says that when $d \gg m$, a wrong key has $(1-m/d)$ of its energy outside the subspace on average---an \emph{energetic} separation between in-span and out-of-span keys exists. The proposition does \textbf{not} assert that a trained network will necessarily exploit this energy difference to distinguish wrong keys. If the downstream map effectively ignores directions aligned with $(I-P)$, then wrong-key and no-key paths may approach authorized behavior.

\subsection{Random training coefficients: learning the subspace, not a single point}
Resampling $\alpha$ at every training step exposes the network to a \emph{family} of injection directions that span $\mathrm{Span}(B)$, rather than a single fixed perturbation that biases and subsequent layers could learn to cancel. This encourages representations and readout to remain roughly stable across all legitimate keys while still behaving differently from typical out-of-span keys---provided the Jacobian-margin pipeline retains sensitivity to $(I-P)k$.

\subsection{Key absorption}
\label{sec:key_absorption}
\emph{Key absorption} refers to the following phenomenon: under correct-key-only supervision, optimization---through batch normalization, residual connections, scale matching, and related mechanisms---reduces the \emph{effective sensitivity} of the downstream classification margin to the injected key direction. Equivalently, the vector $u_{y,c} := J_\ell^\top(w_y - w_c)$ tends toward small norms under the training-induced distribution. The wrong key does not necessarily vanish from the activations, but its impact on the decision margin can be suppressed, making it difficult to observe strong separation across the three key regimes (no / correct / wrong). This phenomenon is \emph{training-induced} and is the central concept for understanding the baseline experimental results.

\subsection{Proposition: wrong-key margin-flip upper bound}
\label{sec:margin_prop}
Fix $\theta$, input $x$, true class $y$, and competing class $c \neq y$. Assume the reference forward pass satisfies $M_{y,c} > 0$ and that $k \sim \mathcal{N}(0, I_{d_\ell})$, analyzing only the linearized margin.

\begin{proposition}[First-order margin flip under Gaussian wrong key]
\label{prop:margin_tail}
Let $u_{y,c} \neq 0$ and write $\sigma := \|u_{y,c}\|_2$. Then
\[
\Pr\bigl[M_{y,c} + \gamma\, u_{y,c}^\top k < 0\bigr]
= \Phi\!\left(-\frac{M_{y,c}}{\gamma\sigma}\right)
\;\le\;
\exp\!\left(-\frac{M_{y,c}^2}{2\gamma^2 \sigma^2}\right),
\]
where $\Phi$ is the standard normal CDF.
\end{proposition}

The core message of Proposition~\ref{prop:margin_tail}: for a fixed positive margin $M_{y,c} > 0$, the flip probability decays \emph{exponentially} as $\sigma$ decreases. Key absorption shrinks the effective $\sigma$ on typical forward passes, so a random wrong key is unlikely to flip the decision under the linearized approximation---this provides a first-order diagnostic framework for why the three accuracy regimes can converge in baseline experiments. Conversely, applying gradient signals to unauthorized forward passes through a deny loss can prevent $\sigma$ from being driven to negligible values.

\textbf{Limitations of this section.} The discussion above is local, first-order, input-dependent, and does not model optimization dynamics. It shows that separation is not absurd at the linearized-margin level---not that any particular trained network will exhibit strong key dependence. The experiments below provide direct empirical tests.

\subsection{Reject-margin convergence analysis}
\label{sec:reject_margin}

The preceding propositions do not involve Mode~B's distinctive reject-dimension structure. This section characterizes the reject behavior through a scalar \emph{reject margin} and derives its scaling relationships with training steps, number of classes, and deny-loss weight.

\paragraph{Definition.} Under Mode~B, the classification head outputs $z \in \mathbb{R}^{C+1}$; denote the reject-dimension logit by $z_{C+1}$. The reject margin is defined as
\[
\Delta(x) := z_{C+1} - \max_{c \in \{0,\dots,C-1\}} z_c .
\]
When $\Delta > 0$, the reject dimension wins the $(C{+}1)$-way argmax---the model judges the input as unauthorized. For correct-key inputs, the cross-entropy simultaneously pushes $z_{C+1}$ down and the true class $z_y$ up, keeping $\Delta < 0$; for wrong-key inputs, the deny loss provides a steady positive gradient to $z_{C+1}$ alone, driving $\Delta \to +\infty$.

\paragraph{Gradient dynamics and scaling bounds.} Fix one input from a wrong-key batch, write $p_i = \mathrm{softmax}(z)_i$. The Mode~B deny loss is
\[
\mathcal{L}_{\mathrm{deny}}^{(B)} = -\log p_{C+1} = -z_{C+1} + \log\!\Bigl(\sum_{i=0}^{C} e^{z_i}\Bigr).
\]
The gradient components are
\[
\frac{\partial \mathcal{L}}{\partial z_{C+1}} = -(1-p_{C+1}) = -p_{\mathrm{sem}},\qquad
\frac{\partial \mathcal{L}}{\partial z_c} = p_c\;\;(c \le C),
\]
where $p_{\mathrm{sem}} := \sum_{c=0}^{C-1} p_c = 1-p_{C+1}$. Following the negative gradient, $z_{C+1}$ receives an increment of $\eta(1-p_{C+1})$ per step, while $\max_{c\le C} z_c$ receives at most $\eta \cdot \max_{c\le C} p_c$. Let $\tilde{\eta} = \lambda\rho\eta$ be the effective learning rate ($\lambda$ the deny weight, $\rho$ the fraction of batches using wrong keys). Then
\[
\Delta_{t+1} - \Delta_t \ge \tilde{\eta}\bigl[(1-p_{C+1}) - \max_{c\le C} p_c\bigr] \ge \tilde{\eta}\bigl[p_{\mathrm{sem}} - \max_{c\le C} p_c\bigr] \ge 0.
\]
$\Delta_t$ is monotonically non-decreasing and $p_{\mathrm{sem}}$ is monotonically non-increasing. When $\Delta \gg \log C$,
\[
p_{\mathrm{sem}} = \frac{\sum_{c\le C}e^{z_c}}{e^{z_{C+1}}+\sum_{c\le C}e^{z_c}} \le \frac{C e^{\max_{c\le C} z_c}}{e^{z_{C+1}}+C e^{\max_{c\le C} z_c}} = \frac{C}{e^{\Delta}+C} \le C e^{-\Delta},
\]
and $1 - 2p_{\mathrm{sem}} \ge 1 - 2C e^{-\Delta}$, yielding the first-order difference inequality:
\[
\Delta_{t+1} - \Delta_t \ge \tilde{\eta}\bigl(1 - (C+1)e^{-\Delta_t}\bigr),\quad \Delta_t \gg \log C.
\]

\begin{proposition}[Mode~B reject-margin lower bound and semantic-accuracy upper bound]
\label{prop:reject_margin}
Suppose Mode~B training runs for $T$ steps with effective learning rate $\tilde{\eta} = \lambda\rho\eta > 0$ and initialization $\Delta_0 \approx 0$ (i.e., initially $z_{C+1} \approx \max_{c\le C} z_c$). Then there exists a constant $\tau_0$ (depending on $C$ and $\tilde{\eta}$; the linear approximation is valid for $T \gg \tau_0$) such that, with high probability for wrong-key inputs,
\[
\Delta_T \ge \log\!\left(\frac{\tilde{\eta} (T-\tau_0)}{C+1} + 1\right).
\]
Furthermore, the wrong-key semantic accuracy---the fraction of wrong-key inputs where the model predicts the correct class among $\{0,\dots,C-1\}$---is bounded by
\[
\mathrm{Acc}_{\mathrm{sem}}^{\mathrm{wrong}} \le p_{\mathrm{sem}} \le \frac{C}{e^{\Delta_T}+C} \le \min\!\left(1,\; \frac{C(C+1)}{\tilde{\eta}(T-\tau_0) + (C+1)}\right).
\]
\end{proposition}
\begin{proof}[Derivation sketch]
From the gradient-step form and the monotonicity of $z_{C+1}$, we have $p_{\mathrm{sem}}(t) \le C/(e^{\Delta_t}+C)$ throughout. Substituting this bound into the difference inequality gives $\Delta_{t+1} \ge \Delta_t + \tilde{\eta}\bigl(1 - (C+1)/(e^{\Delta_t} + C)\bigr)$. Once $\Delta_t$ exceeds $\log C$, the $C$ in the denominator becomes negligible, simplifying to $\Delta_{t+1} \ge \Delta_t + \tilde{\eta}\bigl(1 - (C+1)e^{-\Delta_t}\bigr)$. Viewing this as the continuous-time ODE $d\Delta/dt = \tilde{\eta}(1 - (C+1)e^{-\Delta})$, separation of variables and integration yields the logarithmic lower bound. The discrete-time analysis introduces the constant $\tau_0$ to absorb the cumulative error from the initial nonlinear phase. The full derivation appears in Appendix~\ref{sec:appendix_margin}.
\end{proof}

Proposition~\ref{prop:reject_margin} makes two testable scaling predictions. First, wrong-key semantic accuracy decays as $O(C^2 / \tilde{\eta} T)$---more classes make rejection harder (probability mass is spread across more semantic dimensions), while larger deny weight $\lambda$, wrong-key batch fraction $\rho$, and training steps $T$ jointly accelerate convergence. Second, the reject margin grows as $\log T$ rather than linearly, implying that Mode~B can establish strong rejection early in training. This is consistent with the empirical observation in Section~\ref{sec:mode_b} that near-perfect rejection is achieved within 20 epochs on MNIST. It also explains why Mode~B strikes a better trade-off than Mode~A: the latter spreads the deny signal evenly across $C$ dimensions, where no single logit can accumulate a $\log T$-level gradient advantage.

All propositions in this section are simplified analyses that assume independence and asymptotic conditions on the network architecture and training process; quantitative deviations are expected in deep nonlinear backbones. Their role is to provide falsifiable scaling relations and a conceptual framework, not precise numerical predictions.

\section{Experiments}
\label{sec:experiments}

\subsection{Experimental setup and evaluation protocol}
\label{sec:exp_setup}

\textbf{Example suite.} The experiments span nine examples (01--09) ordered by task type and model scale: 01 (synthetic-vector MLP), 02 (MNIST CNN), 03 (Fashion-MNIST ResNet-style), and 04 (CIFAR-10 ResNet-18) are vision classification baselines; 05 (MNIST Mode~B ablation) serves single-factor sweeps; 06 (ViT-Tiny CIFAR-10) provides vision Transformer cross-architecture validation; 07 (GPT-2-small text classification) extends Mode~B to language tasks; 08 (GPT-2 WikiText-103 language modeling) covers generative tasks; 09 (Qwen2.5-0.5B text classification) validates cross-architecture LLM behavior.

\textbf{Three-key evaluation protocol.} Every experiment is evaluated under three forward-pass conditions:
\begin{itemize}\itemsep=0pt
  \item \textbf{No key}: injection is disabled or a zero key is supplied;
  \item \textbf{Correct key}: $k \in \mathrm{Span}(B)$, i.e., $k = \alpha^\top B$;
  \item \textbf{Wrong key}: $k \notin \mathrm{Span}(B)$, typically an isotropic Gaussian sample.
\end{itemize}
Classification tasks report top-1 accuracy; under Mode~B ($C{+}1$ classes) we separately report semantic accuracy (argmax restricted to the first $C$ dimensions, compared against the label) and the fraction of inputs assigned to the reject dimension. The generative task (Example~08) reports perplexity ($\exp(\text{avg NLL})$).

\textbf{Injection setup.} Unless noted otherwise, all experiments use multiplicative injection ($\mathrm{Inject}_{\text{mul}}$) at two injection points: after the embedding layer and after the block at roughly half network depth. CNN/MLP examples implement injection through manually unrolled forward passes; Transformer examples use PyTorch forward hooks to intercept intermediate hidden states. Keys are resampled for every batch during evaluation; metrics are averaged over the full test set.

\subsection{Baseline: correct-key-only supervision}
\label{sec:baseline}

Baseline training minimizes the supervised loss on the correct-key path only ($\lambda = 0$), without any deny term. Table~\ref{tab:baseline_all} summarizes the three-key accuracies for Examples~01--04 under both additive and multiplicative injection.

\begin{table}[t]
\centering
\small
\setlength{\tabcolsep}{3.5pt}
\begin{tabular}{ccllrrrr}
\toprule
Example & Injection & Dataset & Network & No-key & Correct-key & Wrong-key & Random \\
\midrule
\multicolumn{8}{c}{\textbf{Training set}} \\
\midrule
01 & add & Synthetic & MLP & 0.880 & 0.979 & 0.489 & 0.200 \\
01 & mul & Synthetic & MLP & 0.962 & 0.992 & 0.672 & 0.200 \\
02 & add & MNIST & CNN & 0.990 & 0.990 & 0.989 & 0.100 \\
02 & mul & MNIST & CNN & 0.990 & 0.989 & 0.989 & 0.100 \\
03 & add & F-MNIST & ResNet-like & 0.858 & 0.857 & 0.552 & 0.100 \\
03 & mul & F-MNIST & ResNet-like & 0.568 & 0.890 & 0.868 & 0.100 \\
04 & add & CIFAR-10 & ResNet-18 & 0.980 & 0.989 & 0.333 & 0.100 \\
04 & mul & CIFAR-10 & ResNet-18 & 0.831 & 0.996 & 0.995 & 0.100 \\
\midrule
\multicolumn{8}{c}{\textbf{Test / held-out set}} \\
\midrule
01 & add & Synthetic & MLP & 0.708 & 0.750 & 0.558 & 0.200 \\
01 & mul & Synthetic & MLP & 0.833 & 0.825 & 0.638 & 0.200 \\
02 & add & MNIST & CNN & 0.993 & 0.993 & 0.992 & 0.100 \\
02 & mul & MNIST & CNN & 0.993 & 0.992 & 0.991 & 0.100 \\
03 & add & F-MNIST & ResNet-like & 0.851 & 0.858 & 0.572 & 0.100 \\
03 & mul & F-MNIST & ResNet-like & 0.583 & 0.887 & 0.868 & 0.100 \\
04 & add & CIFAR-10 & ResNet-18 & 0.844 & 0.844 & 0.317 & 0.100 \\
04 & mul & CIFAR-10 & ResNet-18 & 0.818 & 0.904 & 0.901 & 0.100 \\
\bottomrule
\end{tabular}
\caption{Three-key accuracies under baseline training (no deny loss). Example~01 uses a held-out 20\% of synthetic samples; 02--04 use standard test sets. The random column is $1/C$ for $C$-class uniform guessing.}
\label{tab:baseline_all}
\end{table}

\paragraph{What the baselines reveal.}
Table~\ref{tab:baseline_all} displays several distinct patterns. On easy tasks (MNIST), the three accuracies nearly coincide---the network can solve the task while largely ignoring the injection signal. On harder tasks (CIFAR-10), \texttt{add} injection yields a markedly lower wrong-key accuracy (test 0.317), while \texttt{mul} injection keeps the three accuracies close together (0.818 / 0.904 / 0.901). Fashion-MNIST under \texttt{mul} shows yet another shape: no-key accuracy (0.583) is substantially below correct-key (0.887) and wrong-key (0.868), suggesting that the no-key path---lacking injection altogether---suffers a covariate shift from the training distribution.

The central observation: under correct-key-only supervision, \emph{there is no a priori monotonic relationship} among the three accuracies. Multiple regimes coexist. The injection type (add vs.\ mul), task difficulty, and network architecture jointly determine the degree of separation. This stands in tension with the naive expectation that ``illegitimate keys should always produce degraded output.'' The next section offers a unified account through the lens of key absorption.

\subsection{Key absorption: diagnosing weak baseline separation}
\label{sec:diagnosis}

Proposition~\ref{prop:beta_energy} shows that a wrong key's energy mostly lies outside the subspace; Proposition~\ref{prop:margin_tail} shows that when the effective sensitivity $\sigma = \|u_{y,c}\|_2$ is small enough, a random wrong key flips the decision margin with exponentially low probability.

\emph{Key absorption} (Section~\ref{sec:key_absorption}) connects this theoretical tension to the empirical baselines. Under correct-key-only supervision, the optimization applies \emph{no} gradient to wrong-key paths. The network can therefore effectively ``ignore'' the injection direction by driving $\|u_{y,c}\|$ to small values---batch-normalization scaling, residual connections' identity shortcuts, and adaptive compression of the injection direction in downstream layers all contribute. The consequence: although wrong keys carry different perturbation energy than correct keys, their impact on the final decision boundary is suppressed, and the three accuracies converge.

\textbf{Why injection type matters.} Under \texttt{add} injection, the perturbation is independent of the current activation and can retain more noticeable margin differences on harder tasks like CIFAR-10 (wrong-key test accuracy 0.317 vs.\ correct-key 0.844). Under \texttt{mul} injection, the key couples with the activation through $h \odot \tanh(k)$, giving the network more ``absorption'' degrees of freedom---it can weaken the net effect of the multiplicative factor by modulating activation magnitudes. This is why \texttt{mul} more readily leads to converging accuracies (0.818 / 0.904 / 0.901).

\textbf{The fix.} The root cause of key absorption is the absence of gradient signals on unauthorized forward passes. Introducing explicit deny losses---additional supervision on wrong-key and optionally no-key paths---prevents the optimization from compressing $\sigma$ without cost. The next section details this approach.

\subsection{The deny-loss family: Modes A--C and extensions}
\label{sec:deny_family}

Section~\ref{sec:diagnosis} showed that margin sensitivity to the key direction can be absorbed when only the correct-key path is supervised. The direct remedy is to add supervision on unauthorized forward passes. The experiments below all use Example~04 (CIFAR-10 ResNet-18), which exhibited the most pronounced absorption--separation tension among the baselines.

\paragraph{Training objective.}
On each batch, the standard CE loss is minimized on the authorized (correct-key) forward pass, and a deny term is added on selected unauthorized forward passes (wrong-key, and optionally no-key), as specified in Eq.~\eqref{eq:total_loss}.

\subsubsection{Classical Modes A, B, C}

\textbf{Mode A: entropy maximization.} Let $z \in \mathbb{R}^C$ be the logits from an unauthorized forward pass, $p = \mathrm{softmax}(z)$, and let $H(p) = -\sum_c p_c \log p_c$ be the Shannon entropy. Mode~A pushes unauthorized outputs toward high entropy:
\[
\mathcal{L}_{\mathrm{deny}}^{(A)} = -\mathbb{E}\bigl[H(p)\bigr].
\]

\textbf{Mode B: explicit reject class.} The classification head is expanded from $C$ to $C{+}1$ classes, with the $(C{+}1)$-th dimension serving as a reject pseudo-class. Unauthorized forward passes are supervised to predict the reject class:
\[
\mathcal{L}_{\mathrm{deny}}^{(B)} = \mathbb{E}\bigl[\ell_{\mathrm{CE}}(z,\, r)\bigr],\quad r = C{+}1.
\]
Authorized forward passes are still supervised on the true label within $\{1{:}C\}$. At evaluation time, the $(C{+}1)$-way argmax is taken; semantic accuracy compares only the top prediction among the first $C$ dimensions against the label. If the reject dimension wins, the semantic prediction is counted as incorrect.

\textbf{Mode C: margin hinge.} Let $z^{\mathrm{ok}}$ and $z^{\mathrm{w}}$ be the logits for the same batch under correct and wrong keys, respectively:
\[
\mathcal{L}_{\mathrm{deny}}^{(C)} = \mathbb{E}\Bigl[\max\bigl(0,\; m - (z^{\mathrm{ok}}_{y} - z^{\mathrm{w}}_{y})\bigr)\Bigr],
\]
where $m > 0$ is a margin hyperparameter.

Table~\ref{tab:deny_abc} summarizes the comparison of Modes A--C on Example~04 ($\lambda = 0.1$, deny loss applied only to the wrong-key path).

\begin{table}[t]
\centering
\small
\setlength{\tabcolsep}{4pt}
\begin{tabular}{lrrrr}
\toprule
Variant & No-key & Correct-key & Wrong-key & Random \\
\midrule
\multicolumn{5}{c}{\textbf{Training set}} \\
\midrule
Baseline (no $\mathcal{L}_{\mathrm{deny}}$) & 0.847 & 0.991 & 0.972 & 0.100 \\
Mode A & 0.407 & 0.827 & 0.642 & 0.100 \\
Mode B & 0.000 & 0.991 & 0.000 & 0.100 \\
Mode C ($m{=}1.0$) & 0.785 & 0.990 & 0.974 & 0.100 \\
\midrule
\multicolumn{5}{c}{\textbf{Test set}} \\
\midrule
Baseline (no $\mathcal{L}_{\mathrm{deny}}$) & 0.827 & 0.898 & 0.885 & 0.100 \\
Mode A & 0.478 & 0.792 & 0.644 & 0.100 \\
Mode B & 0.000 & 0.903 & 0.000 & 0.100 \\
Mode C ($m{=}1.0$) & 0.771 & 0.901 & 0.888 & 0.100 \\
\bottomrule
\end{tabular}
\caption{Modes A--C comparison on Example~04: semantic accuracy. The zero values for Mode~B under no-key and wrong-key reflect the reject dimension winning the $(C{+}1)$-way argmax.}
\label{tab:deny_abc}
\end{table}

\paragraph{How the three modes differ.}
Table~\ref{tab:deny_abc} reveals distinct trade-offs between correct-key generalization and unauthorized degradation. Mode~A's difficulty stems from a direct conflict: pulling unauthorized outputs toward a uniform distribution fights against the sharp peaks demanded by the correct-key cross-entropy, and the two objectives tug on a shared backbone. Correct-key accuracy drops markedly from the baseline 0.898 to 0.792, while unauthorized degradation is incomplete (wrong-key still at 0.644). Mode~B sidesteps this conflict. The deny objective is concentrated on a single added dimension, leaving the correct-key path undisturbed (correct-key 0.903, close to baseline) while unauthorized outputs collapse to near-zero semantic accuracy. Mode~C applies a hinge constraint only to the true-class coordinate; unauthorized outputs can retain high confidence on other semantic dimensions, so wrong-key accuracy barely budges from the baseline (0.888)---its deny signal is too weak to meaningfully reshape the unauthorized forward-pass behavior.

\paragraph{The subspace structure is necessary: a non-subspace Gaussian comparison.}
Mode~B with subspace keys achieves near-binary gating (correct 0.903 vs.\ wrong 0.000), but this result does not by itself prove that the subspace structure is necessary. To test this claim directly, we ran a control experiment under identical architecture, injection type, and training protocol, with one change: keys were drawn from a pure random Gaussian ($k \sim \mathcal{N}(0, I_d)$) instead of the subspace construction ($k = \alpha^{\top}B$). In this control, ``correct'' and ``wrong'' keys both come from the same distribution---the only difference is that they are two independent Gaussian samples, with no structural in-span vs.\ out-of-span relationship.

\begin{table}[H]
\centering
\small
\setlength{\tabcolsep}{7pt}
\begin{tabular}{lrrrr}
\toprule
\textbf{Method} & \textbf{No-key} & \textbf{Correct-key} & \textbf{Wrong-key} & \textbf{Reject rate (wrong)} \\
\midrule
SpanKey Mode B (subspace) & 0.000 & 0.903 & 0.000 & 1.000 \\
Non-subspace Gaussian (control) & 0.814 & 0.895 & 0.894 & 0.000 \\
\bottomrule
\end{tabular}
\caption{Subspace vs.\ non-subspace Gaussian keys on Example~04 (test set, 100-epoch Mode~B). Under the non-subspace control, the three key regimes are nearly indistinguishable, and the reject dimension never activates.}
\label{tab:nonsubspace}
\end{table}

Table~\ref{tab:nonsubspace} is decisive. In the control experiment, correct-key (0.895) and wrong-key (0.894) semantic accuracies are essentially identical---the model cannot build a consistent differentiation between them because both sets of keys are drawn from the same Gaussian and no structural in-span/out-of-span distinction exists to be learned. More tellingly, the reject dimension is never predicted (wrong-key reject rate 0.000), meaning Mode~B's $(C{+}1)$-dimensional supervision completely fails without the subspace structure: the model has no reason to route any Gaussian sample to the reject dimension, since all samples are distributionally equivalent. The no-key accuracy of 0.814 further indicates that the model learned to partially ignore the injection perturbation---consistent with the key absorption mechanism described in Section~\ref{sec:key_absorption}---but since absorption acts identically on correct and wrong keys in the absence of subspace structure, the two cannot separate.

This control establishes SpanKey's design logic from the negative side: the subspace $B$ provides a low-dimensional control channel that allows optimization to carve ``inside the subspace'' vs.\ ``outside the subspace'' into a learnable class boundary. Random high-dimensional Gaussian perturbations alone do not constitute a workable gating primitive---the model treats all such perturbations as equivalent noise and learns to ignore them, rather than learning to distinguish authorized from unauthorized forward passes.

\subsubsection{Extended variants}

To mitigate Mode~A's optimization conflict or strengthen Mode~C's constraints, we introduce the following extensions (Table~\ref{tab:deny_ext}, 40 epochs, MultiStepLR schedule proportional to the main-text baseline, $\lambda = 0.1$):

\textbf{A\_soft:} A quadratic penalty is applied only when $H(p) < \log C - \eta$, so that already-uniform unauthorized outputs receive no further deny gradient---effectively sparsifying the deny signal and reducing the zero-sum tug-of-war with the primary loss.

\textbf{cplus:} In addition to Mode~C, a hinge penalty is applied to the maximum non-true-class logit under the wrong key, closing the shortcut where ``the wrong key still predicts the true class.''

\textbf{AC:} $\mathcal{L}_{\mathrm{AC}} = \frac12 \mathcal{L}_{\mathrm{Asoft}} + \frac12 \mathcal{L}_{\mathrm{deny}}^{(C)}$, a convex combination that prevents either extreme objective from dominating the representation.

\textbf{B\_aux:} The $C$-dimensional semantic head is kept, and a separate scalar auxiliary logit $r = \phi(h_L)$ is added, passed through a sigmoid to produce a reject probability supervised with binary cross-entropy. Semantic classification and rejection are handled by separate channels---a task-decomposition approach.

\begin{table}[t]
\centering
\small
\setlength{\tabcolsep}{4pt}
\begin{tabular}{lrrr}
\toprule
Variant & No-key & Correct-key & Wrong-key \\
\midrule
\multicolumn{4}{c}{\textbf{Training set}} \\
\midrule
Baseline & 0.830 & 0.939 & 0.915 \\
A\_soft & 0.498 & 0.888 & 0.697 \\
cplus & 0.435 & 0.920 & 0.535 \\
AC & 0.642 & 0.905 & 0.778 \\
B\_aux & 0.784 & 0.937 & 0.912 \\
\midrule
\multicolumn{4}{c}{\textbf{Test set}} \\
\midrule
Baseline & 0.818 & 0.883 & 0.859 \\
A\_soft & 0.547 & 0.852 & 0.684 \\
cplus & 0.441 & 0.865 & 0.509 \\
AC & 0.679 & 0.860 & 0.763 \\
B\_aux & 0.786 & 0.882 & 0.862 \\
\bottomrule
\end{tabular}
\caption{Extended-variant comparison on Example~04 (40 epochs, $\lambda{=}0.1$). For B\_aux, the headline is the $C$-class argmax; the reject signal is independent, via $\sigma(r)$. In this run, $\sigma(r) \approx 0.98$ on unauthorized forward passes and $\approx 0.006$ on correct-key passes.}
\label{tab:deny_ext}
\end{table}

\paragraph{Patterns across the extended variants.}
Table~\ref{tab:deny_ext} surfaces several instructive trends. A\_soft's sparsified deny strategy achieves moderate unauthorized degradation (wrong-key 0.684) at almost no cost to correct-key accuracy (0.852), suggesting that not all unauthorized samples need equally strong deny gradients---a finding with immediate implications for deployment efficiency. cplus imposes the most structured multi-coordinate pressure and achieves the most thorough unauthorized degradation (0.509), at a slight cost to correct-key accuracy. AC strikes a continuous compromise between the two. B\_aux adopts a different philosophy: it fully decouples the reject signal from the semantic channel, delegating it to an independent scalar head. Semantic top-1 (0.882) is barely affected, while $\sigma(r)$ gives approximately 0.98 reject probability on unauthorized forward passes and approximately 0.006 on correct-key passes---an almost ideal binary split. The common thread: any mechanism that applies directional gradients to unauthorized forward passes fundamentally prevents the optimization from driving the margin sensitivity vector $u_{y,c}$ to negligible values (see Proposition~\ref{prop:margin_tail}); the modes differ only in efficiency and side effects.

\subsection{Mode~B mechanism and ablation}
\label{sec:mode_b}

\subsubsection{Why Mode~B stands out}

The ``high correct-key accuracy + strong unauthorized rejection'' pattern that Mode~B exhibits in Table~\ref{tab:deny_abc} is not accidental. Its advantage is rooted in the structure of the logit space.

Unlike Mode~A---where the deny objective forces the $C$-dimensional softmax toward uniformity, directly colliding with the correct-key path's need for sharp peaks on a shared backbone---Mode~B offloads the deny task onto a single added dimension $z_{C+1}$. An unauthorized forward pass satisfies the deny loss simply by pushing $z_{C+1}$ far above the semantic logits, without needing to rewrite the relative ordering among $\{z_1,\dots,z_C\}$. The correct-key path never receives supervisory signal on $z_{C+1}$ at all; it competes only on the true class within the first $C$ dimensions. This \emph{partial decoupling of tasks across logit dimensions} is the structural reason Mode~B outperforms Mode~A.

Proposition~\ref{prop:reject_margin} adds a second layer of explanation from gradient dynamics. On the wrong-key path, $z_{C+1}$ receives a steady positive increment at every step, while the maximum among the semantic logits receives only a push of order $O(e^{-\Delta_t})$. This gradient asymmetry drives the reject margin $\Delta_t$ monotonically upward---at a $\log T$ rate, meaning strong rejection is already established early in training. The contrast with Mode~A is instructive: there, the deny signal is diluted across $C$ dimensions, none of which can accumulate a $\log T$-level gradient advantage.

When the no-key forward pass is also included as an unauthorized target, the model unifies ``key absent'' and ``wrong key present'' under the same reject response. A near-zero no-key semantic accuracy at test time reflects this gating design, not a loss of the network's ability to output semantic classes.

\subsubsection{MNIST Mode~B ablation (Example~05)}

To supplement the mechanistic account with a scannable design-space profile, we conduct a single-factor ablation of Mode~B on an MNIST small CNN (Example~05). An anchor configuration is fixed ($m{=}8$, inject layers $\{0,2\}$, $\gamma{=}0.5$, \texttt{mul}, $\lambda{=}0.1$, deny on wrong-key only), and one factor is varied at a time. Figure~\ref{fig:ablation} reports test-set semantic accuracy and wrong-key reject-dimension fraction across three sweeps: (A)~$m \in \{4,8,16,32\}$; (B)~injection layer combinations; (C)~$\gamma \in \{0.25,0.5,1.0,2.0\}$.

\begin{figure}[htbp]
\centering
\includegraphics[width=\linewidth]{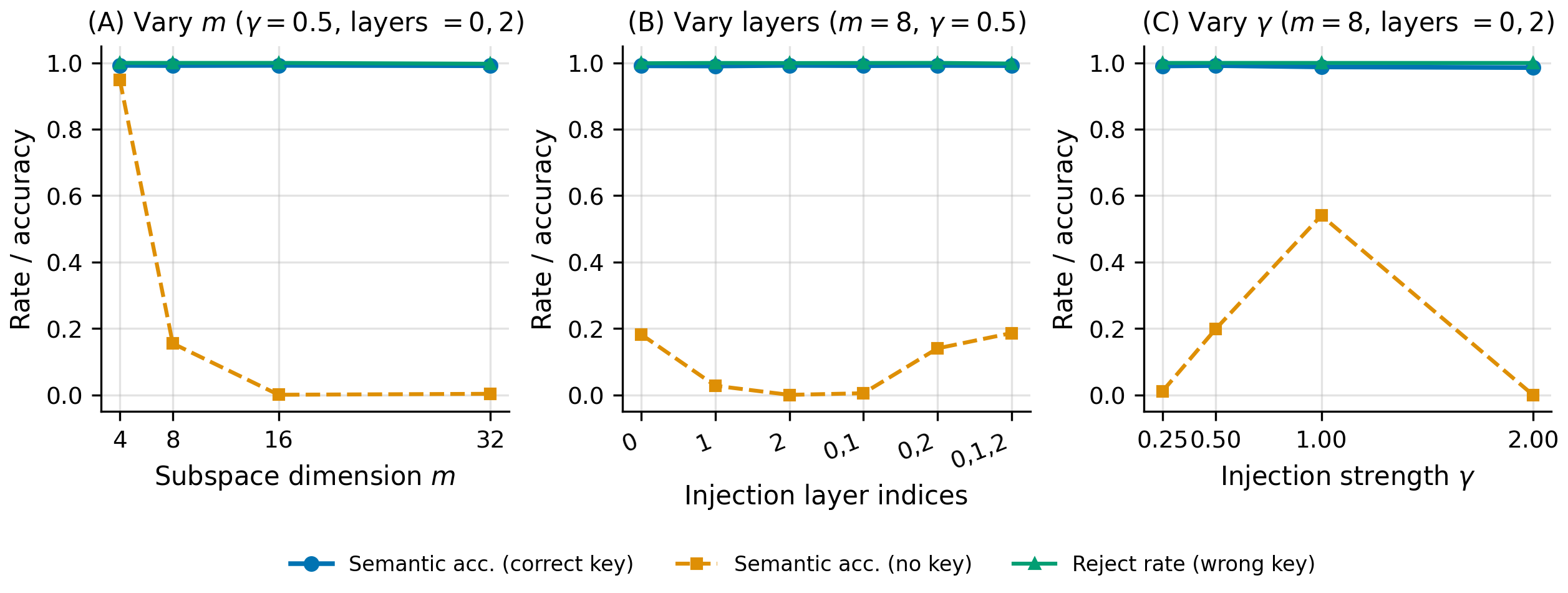}
\caption{MNIST Example~05 Mode~B ablation: test-set semantic accuracy (correct-key, no-key) and wrong-key reject fraction. The wrong-key reject rate is near 1 in most configurations; no-key semantic accuracy is more sensitive to key resampling stochasticity during evaluation---attend to trends and relative comparisons.}
\label{fig:ablation}
\end{figure}

The main findings: (1)~The wrong-key reject fraction is close to 1 across most configurations---consistent with Proposition~\ref{prop:reject_margin}'s $\log T$ convergence prediction: 20 epochs suffice to push the reject margin well past $\log C$. (2)~Correct-key semantic accuracy hovers in the 0.98--0.99 range with minimal sensitivity to the design knobs. (3)~No-key semantic accuracy varies non-monotonically with $m$, $\gamma$, and layer placement; larger $\gamma$ can drive no-key semantics to near zero, reflecting the trade-off between strong conditioning and no-key usability.

\subsection{Cross-architecture validation: ViT-Tiny (Example~06)}
\label{sec:vit}

To test whether SpanKey transfers to non-convolutional vision backbones, we run a ViT-Tiny (timm \texttt{vit\_tiny\_patch16\_224}) Mode~B experiment on CIFAR-10. Injection points are after patch encoding and after the 6th block (of 12 total), \texttt{mul}, $m{=}16$, $\lambda{=}0.1$ (wrong-key only), per-layer independent coefficients. Training runs for 20 epochs; results appear in Table~\ref{tab:vit}.

\begin{table}[H]
\centering
\small
\setlength{\tabcolsep}{5pt}
\begin{tabular}{lcccc}
\toprule
\textbf{CIFAR-10 test set} & \textbf{No-key sem.} & \textbf{Correct-key sem.} & \textbf{Wrong-key sem.} & \textbf{Wrong-key rej. frac.} \\
\midrule
ViT-Tiny, Mode B & 0.000 & 0.755 & 0.000 & 1.000 \\
\bottomrule
\end{tabular}
\caption{ViT-Tiny CIFAR-10 Mode~B: semantic accuracy and wrong-key reject fraction.}
\label{tab:vit}
\end{table}

SpanKey assembles an explicit reject dimension on Transformer-style visual representations with the reject behavior fully intact (wrong-key rejection 1.000). The correct-key semantic accuracy of 75.5\% needs to be read in context. ViT-Tiny's capacity on CIFAR-10 is constrained from two sides: Transformers lack the translation equivariance and locality inductive biases of CNNs and need longer training to converge on 50k images; moreover, roughly half the optimization steps in Mode~B dual-objective training go toward deny losses rather than semantic classification. For reference, a standard ViT-Tiny under the same short training budget (20 epochs) without SpanKey is expected to reach roughly 80--84\%---the 5--10 percentage-point gap reflects the reasonable capacity overhead of the gating task. The reject behavior reaches 100\% within 20 epochs, substantially faster than semantic classification converges, which aligns with Proposition~\ref{prop:reject_margin}'s $\log T$ prediction: the gating signal only needs to accumulate a gradient advantage on the reject-dimension logit, not to build full decision boundaries across all semantic dimensions.

\subsection{Text classification: GPT-2 Mode~B (Example~07)}
\label{sec:gpt2_cls}

To move SpanKey from vision to language tasks, we implement Mode~B text classification on GPT-2-small (124M). Injection points are after the dropout following word-plus-position embeddings, and after the output of the roughly mid-depth block ($h[5]$ of 12 layers). The classification head takes the hidden state of the last effective token through a linear map to $C{+}1$ dimensions. Training uses multiplicative injection, $\gamma{=}0.15$, $m{=}16$, $\lambda{=}0.1$ (wrong-key only), per-layer independent coefficients.

Results on SST-2 (binary) and AG News (4-way) after 5 epochs are shown in Table~\ref{tab:gpt2_cls}.

\begin{table}[H]
\centering
\small
\setlength{\tabcolsep}{4pt}
\begin{tabular}{lcccc}
\toprule
\textbf{Task} & \textbf{No-key sem.} & \textbf{Correct-key sem.} & \textbf{Wrong-key sem.} & \textbf{Wrong-key rej. frac.} \\
\midrule
SST-2 (test) & 0.000 & 0.901 & 0.000 & 1.000 \\
AG News (test) & 0.000 & 0.895 & 0.000 & 1.000 \\
\bottomrule
\end{tabular}
\caption{GPT-2-small text classification Mode~B. Test-set semantic accuracy and wrong-key reject fraction.}
\label{tab:gpt2_cls}
\end{table}

Mode~B on GPT-2 reproduces the same ``high correct-key accuracy + strong unauthorized rejection'' pattern observed in vision tasks, confirming that SpanKey operates effectively in language-model hidden-state spaces.

\subsection{Generative extension: GPT-2 language modeling (Example~08)}
\label{sec:gpt2_lm}

\paragraph{Motivation: from classification to generation.}
All experiments so far are classification tasks, where Mode~B relies on a $(C{+}1)$-dimensional reject class. Generative tasks such as language modeling have no classification head---the model outputs a distribution over a vocabulary, and the evaluation metric is perplexity rather than accuracy. This experiment tests whether SpanKey's key-conditioning mechanism transfers to generative settings.

\paragraph{What changes in the injection pipeline.}
The core difference between classification and generation lies in the \emph{output head and reject mechanism}:

\begin{itemize}\itemsep=2pt
  \item \textbf{Classification (Mode~B):} Hidden-state injection $\to$ last-token pooling $\to$ $(C{+}1)$-way linear classifier. The deny loss is cross-entropy toward the $(C{+}1)$-th class on wrong-key / no-key paths---routing unauthorized forward passes to an explicit reject dimension.
  \item \textbf{Generation (LM):} Hidden-state injection $\to$ LM head directly outputs vocabulary logits (no intermediate pooling). The reject mechanism uses \emph{entropy maximization} (Mode~A style): the output entropy on no-key forward passes is \emph{maximized} (i.e., negative entropy is minimized), forcing near-uniform output distributions and producing orders-of-magnitude PPL gaps. The training objective is:
  \[
  \mathcal{L} = \mathcal{L}_{\mathrm{LM}}^{\mathrm{ok}} - \lambda \cdot \mathbb{E}_{\text{no-key}}\bigl[H(p)\bigr],
  \]
  where $\mathcal{L}_{\mathrm{LM}}^{\mathrm{ok}}$ is the standard language-modeling negative log-likelihood on the correct-key path.
\end{itemize}

\emph{The injection operation itself is identical across task types}---a multiplicative perturbation $h \odot (1 + \gamma \tanh(k))$ applied to selected Transformer-layer hidden states. Only the output-side evaluation and the deny target differ. SpanKey's core claim---``subspace key conditioning shapes model behavior''---is thus independent of the specific output-head form.

\paragraph{Setup.}
GPT-2-small (124M) is trained on WikiText-103 for 3 epochs, $\gamma{=}1.0$, $m{=}16$, $\lambda{=}0.1$ (entropy-maximization deny loss, no-key path only, every other batch), per-layer independent coefficients, injection after dropout and after $h[5]$. Perplexity is computed over the full test set under all three key regimes.

\paragraph{Results.}
Table~\ref{tab:gpt2_lm} reports PPL and per-token loss across $\gamma$ values. The \textbf{overall trend} is clear: as $\gamma$ increases from 0.01 to 1.0, the key-conditioning effect strengthens sharply; at $\gamma{=}2.0$, correct-key PPL begins to rise noticeably, signaling that over-injection starts to hurt primary task quality.

\begin{table}[H]
\centering
\small
\setlength{\tabcolsep}{6pt}
\begin{tabular}{lrrrr}
\toprule
$\boldsymbol{\gamma}$ & \textbf{No-key PPL} & \textbf{Correct-key PPL} & \textbf{Wrong-key PPL} & \textbf{Correct / wrong ratio} \\
\midrule
0.01 & 13.87 & 13.87 & 13.87 & 1.00$\times$ \\
0.5  & 40,864 & 13.99 & 39,579 & 2,829$\times$ \\
1.0  & 42,263 & 14.77 & 38,358 & 2,597$\times$ \\
2.0  & 42,248 & 17.39 & 20,194 & 1,161$\times$ \\
\bottomrule
\end{tabular}
\caption{GPT-2 (124M) WikiText-103 language modeling: three-key PPL under a $\gamma$ sweep. All settings other than $\gamma$ are held constant ($m{=}16$, $\lambda{=}0.1$ entropy-maximization deny, per-layer independent coefficients, 3 epochs).}
\label{tab:gpt2_lm}
\end{table}

\textbf{The effective range of $\gamma$.} At $\gamma{=}0.01$, the three PPLs completely overlap (all $\approx$13.87)---the injection perturbation is so weak that the model can ignore it entirely. Key conditioning has a \emph{threshold}: below it, neither in-span nor out-of-span perturbations constitute a distinguishable signal. At $\gamma{=}0.5$, strong key differentiation is already in place: no-key and wrong-key PPL rise to around 40,000 ($\sim$2,800$\times$), while correct-key PPL at 13.99 is barely degraded, indicating that this strength has crossed the threshold without overshooting. $\gamma{=}1.0$ represents the best balance between differentiation and correct-key quality---correct PPL 14.77, unauthorized PPL in the 38,000--42,000 range. At $\gamma{=}2.0$, \emph{over-injection} sets in: correct-key PPL rises from 14.77 to 17.39 (primary task quality drops), while wrong-key PPL actually \emph{falls} to 20,194 (unauthorized degradation eases). The gap between them narrows from $\sim$2,600$\times$ to $\sim$1,200$\times$. This reversal suggests that excessively strong multiplicative perturbations begin to distort normal representations; the optimization reallocates capacity between the semantic and gating tasks, partially offsetting the entropy-maximization deny effect.

\textbf{Generative vs.\ classification Mode~B.} Classification Mode~B produces near-binary reject behavior (authorized / unauthorized semantic accuracy 0.903 / 0.000); the generative setting produces a continuous PPL gap (14.77 / 42,263). The difference traces back to the reject mechanisms: the classification reject dimension is a crisp hard decision boundary ($C{+}1$-way argmax), whereas entropy maximization acts on the overall shape of the vocabulary distribution, manifesting as high PPL rather than hard rejection. In addition, the generative experiment applies the deny loss only to the no-key path (the wrong-key path receives no direct supervision), so it is expected that wrong-key PPL sits slightly below no-key PPL.

\paragraph{Takeaway.}
This experiment confirms that SpanKey's key-conditioning mechanism is \emph{not restricted to classification}. Changing only the output-side objective (CE $\to$ NLL) and the deny target (reject class $\to$ entropy maximization), the same subspace injection primitives produce a three-order-of-magnitude PPL gap on a generative task. The $\gamma$ sweep further reveals a full phase transition from ``ineffective'' (0.01) through ``optimal differentiation'' (0.5--1.0) to ``over-injection'' (2.0), providing quantitative guidance for deployment parameter selection.

\subsection{Cross-architecture LLM validation: Qwen2.5 Mode~B (Example~09)}
\label{sec:qwen25}

\paragraph{Motivation: architecture independence.}
The Transformer experiments so far (06--08) all use the GPT-2 architecture (learned position embeddings + GELU + LayerNorm). Modern LLMs adopt different components: RoPE rotary position encoding, SwiGLU activations, RMSNorm. To test SpanKey's sensitivity to these architectural choices, we reproduce Mode~B text classification on Qwen2.5-0.5B and compare directly with Example~07 (GPT-2).

\paragraph{Injection-point adaptation.}
The architectural differences between GPT-2 and Qwen2.5 require adjusting hook registration points, not the injection logic itself:

\begin{itemize}\itemsep=0pt
  \item \textbf{GPT-2:} Word embeddings (\texttt{wte}) and position embeddings (\texttt{wpe}) are summed and passed through \texttt{drop}. Injection point 1 is the output hook of \texttt{transformer.drop}; injection point 2 is the output hook of \texttt{transformer.h[mid]}.
  \item \textbf{Qwen2.5:} Uses \texttt{embed\_tokens} (a single embedding layer; RoPE is built into the attention modules), with no separate \texttt{wpe} or \texttt{drop}. Injection point 1 is the output hook of \texttt{model.embed\_tokens}; injection point 2 is the output hook of \texttt{model.layers[mid]}.
\end{itemize}

Additionally, Qwen2.5-0.5B loads in bfloat16 by default, while key vectors are float32. Explicit dtype alignment ($k \to$ \texttt{x.dtype}) is needed at injection time, and the classification head must be cast to match the encoder dtype. These engineering details do not affect the mechanism itself but illustrate the implementation differences that cross-architecture deployment must address.

\paragraph{Setup.}
Qwen2.5-0.5B is trained on SST-2 for 2 epochs, $\gamma{=}0.15$, $m{=}16$, $\lambda{=}0.05$ (wrong-key only), per-layer independent coefficients, injection after \texttt{embed\_tokens} and after \texttt{layers[11]} (of 24 total). Gradient clipping ($\max\_norm{=}1.0$) stabilizes the joint optimization of CE and deny loss.

\paragraph{Results and GPT-2 comparison.}
Table~\ref{tab:qwen25} reports Qwen2.5-0.5B results on the SST-2 test set alongside GPT-2-small (Example~07).

\begin{table}[H]
\centering
\small
\setlength{\tabcolsep}{4pt}
\begin{tabular}{lcccc}
\toprule
\textbf{Model (SST-2 test)} & \textbf{No-key sem.} & \textbf{Correct-key sem.} & \textbf{Wrong-key sem.} & \textbf{Wrong-key rej. frac.} \\
\midrule
GPT-2-small (124M, Example~07) & 0.000 & 0.901 & 0.000 & 1.000 \\
Qwen2.5-0.5B (Example~09) & 0.000 & 0.945 & 0.000 & 1.000 \\
\bottomrule
\end{tabular}
\caption{Qwen2.5-0.5B vs.\ GPT-2-small Mode~B on SST-2.}
\label{tab:qwen25}
\end{table}

Qwen2.5-0.5B achieves higher correct-key semantic accuracy (94.5\%) than GPT-2-small (90.1\%), attributable to larger model capacity and stronger pretrained representations, while maintaining a perfect unauthorized reject rate (100\%). The result indicates that SpanKey's subspace injection + Mode~B mechanism is insensitive to the specific components of the underlying Transformer architecture (RoPE vs.\ learned PE, SwiGLU vs.\ GELU, RMSNorm vs.\ LayerNorm)---as long as hookable intermediate hidden states exist, injection can proceed.

\paragraph{Training stability note.}
In preliminary runs with default parameters ($\lambda{=}0.1$, $\gamma{=}0.25$), the deny loss spiked while CE simultaneously deteriorated, reflecting gradient conflict between the two objectives amplified through intermediate layers. Reducing $\lambda$ to 0.05, $\gamma$ to 0.15, and introducing gradient clipping stabilized training throughout. This matches the empirical pattern noted in Section~\ref{sec:deny_family}: deny weight and injection strength must be jointly tuned to balance correct-key generalization against unauthorized degradation.

\section{Threat model and security discussion}
\label{sec:security}

This section delineates the discussion boundaries of SpanKey. Unlike systems aimed at encrypted or multi-party collaborative training/inference~\citep{secureml,homomorphic}, this work does \textbf{not} provide confidential inference or unforgeability guarantees in a cryptographic sense. We examine attack scenarios in two layers: (1)~classical probe verification---the ineffectiveness of $\alpha$ search when $B$ is known; (2)~realistic white-box attacks---feasible bypass paths available to an attacker who holds the model weights but does not know $B$.

\subsection{Threat assumptions}
\begin{itemize}\itemsep=0pt
  \item \textbf{Secrets and interface:} The basis matrix $B$, the injection specification ($\gamma$, layer placement) are considered secret at deployment. The attacker may obtain the \textbf{model weights} (white-box leakage) but not the above secrets.
  \item \textbf{Query budget:} The attacker is subject to rate limiting and account binding; unbounded queries could support high-dimensional subspace estimation.
  \item \textbf{Labels and data:} The attacker may possess a small amount of labeled data drawn from the training distribution.
\end{itemize}

\subsection{Classical probe verification: the futility of $\alpha$ search when $B$ is known}

If the attacker \emph{knows} $B$, $\gamma$, and the injection layers, they only need to search for coefficients $\alpha$ in an $m$-dimensional space. Three probe types are evaluated on Example~05 (MNIST Mode~B, $m{=}8$): in-span random search (300 trials), gradient-based optimization (300 Adam steps), and pure forward queries. The results in Table~\ref{tab:attack} show that the best $\alpha$ found by search or optimization (semantic accuracy 0.992--0.993) is indistinguishable from a single random sample (0.992)---because the training phase already exposed the model to a large set of $\alpha$ values within the subspace. \emph{Any} $\alpha \in \mathbb{R}^m$ produces an equally valid key.

\begin{table}[H]
\centering
\footnotesize
\setlength{\tabcolsep}{5pt}
\begin{tabular}{lcc}
\toprule
\textbf{Setting (MNIST test, $B$ known)} & \textbf{Semantic accuracy} & \textbf{Reject fraction} \\
\midrule
No-key forward & 0.183 & 0.817 \\
Single random in-span key & 0.992 & 0.000 \\
Single random out-of-span key & 0.000 & 1.000 \\
\midrule
Adaptive (300 in-span trials) & 0.993 & 0.000 \\
Black-box query (300 forwards) & 0.992 & 0.000 \\
Gradient attack (300 steps) & 0.993 & 0.000 \\
\bottomrule
\end{tabular}
\caption{In-span probes when $B$ is known. Search and optimization offer no advantage over a single random $\alpha$---all in-span keys are equivalent.}
\label{tab:attack}
\end{table}

This result has a double meaning. On the positive side, the model genuinely learned to accept the entire subspace without overfitting to specific $\alpha$ values. On the negative side, \emph{once $B$ is leaked, the gate collapses}---the attacker need not guess $\alpha$; any value works. SpanKey's security therefore \textbf{depends entirely on the secrecy of $B$}. When the attacker does \emph{not} possess $B$, the real threats come from the bypass paths examined below.

\subsection{Realistic white-box attacks: bypass paths when only weights are known}

The following attacks simulate a more realistic scenario: the attacker obtains the model weights (e.g., through reverse engineering or model repository leakage) but does \emph{not} know $B$, $\gamma$, the injection layers, or the injection type. All experiments use the Example~05 Mode~B model.

\subsubsection{Attack 1: No-injection deployment}

The attacker loads the weights and calls the model's forward pass directly, \textbf{without implementing any injection logic} (for Transformer architectures in Examples~07--09, this is equivalent to registering no forward hooks). Under this scenario, the input passes through without key conditioning, producing a distribution shift from the training-time correct-key-injected forward passes.

\begin{table}[H]
\centering
\small
\setlength{\tabcolsep}{10pt}
\begin{tabular}{lcc}
\toprule
\textbf{Forward-pass method} & \textbf{Semantic accuracy} & \textbf{Reject fraction} \\
\midrule
Direct \texttt{model(x)} (no injection) & 0.009 & 0.992 \\
Correct-key injection (reference; requires $B$) & 0.990 & 0.000 \\
Wrong-key injection (reference; requires $B$) & 0.001 & 0.999 \\
\bottomrule
\end{tabular}
\caption{No-injection deployment: the attacker skips injection logic and calls model.forward(x) directly. The reject fraction is 99.2\%---the model rejects nearly all inputs.}
\label{tab:no_inject}
\end{table}

As Table~\ref{tab:no_inject} shows, the no-injection forward pass triggers a 99.2\% reject fraction. Even though training applied the deny loss only to the wrong-key path (\texttt{deny\_on}{=}wrong, with the no-key path never supervised), the model systematically outputs the reject class in the absence of injection. The reason: the convolutional backbone continuously received injection perturbations during training; no-injection inputs are out-of-distribution for the downstream classification head and are default-routed to the reject dimension.

\subsubsection{Attack 2: Fine-tuning removal of the deny behavior}

The attacker replaces the $(C{+}1)$-dimensional classification head with a $C$-dimensional one (initialized by truncating the first $C$ classes from the original fc weights) and fine-tunes on a small set of labeled samples. Table~\ref{tab:finetune} reports 10-class test accuracy under varying sample budgets and fine-tuning strategies, alongside a from-scratch training baseline.

\begin{table}[H]
\centering
\small
\setlength{\tabcolsep}{6pt}
\begin{tabular}{lccccr}
\toprule
\textbf{Samples} & \textbf{Frozen backbone} & \textbf{Full fine-tune} & \textbf{Scratch} & \textbf{SpanKey pretrain gain} \\
\midrule
50 & 0.979 & 0.979 & 0.648 & +33.1\% \\
100 & 0.987 & 0.987 & 0.723 & +26.4\% \\
500 & 0.988 & 0.989 & 0.889 & +9.9\% \\
1000 & 0.990 & 0.989 & 0.934 & +5.6\% \\
5000 & 0.991 & 0.991 & 0.969 & +2.2\% \\
\bottomrule
\end{tabular}
\caption{Fine-tuning removal of the deny behavior: only 50 labeled samples (5 per class on average) and 10 epochs restore 97.9\% accuracy. Freezing vs.\ full fine-tuning makes essentially no difference---the SpanKey-pretrained backbone retains its feature extraction capability; only the fc layer's reject dimension needs correction.}
\label{tab:finetune}
\end{table}

This is the most realistic threat to SpanKey under weight leakage. Fifty labeled samples (roughly 5 per class) and 10 epochs of fine-tuning restore the model to 97.9\% 10-class accuracy, compared to 64.8\% for training from scratch. Whether the backbone is frozen or not makes almost no difference---the pretrained convolutional features are intact, and only the reject dimension in the fc layer requires correction. This attack requires no knowledge of $B$, $\gamma$, or the injection layers.

\subsubsection{Attack 3: Representation-layer bypass}

The attacker cuts the model at the earliest available intermediate layer after the first injection point (conv1+pool output, $16{\times}14{\times}14{=}3136$ dimensions), discards all subsequent layers, and trains a new classification head (linear probe or MLP) on these frozen features. Results appear in Table~\ref{tab:layer_cut}.

\begin{table}[H]
\centering
\small
\setlength{\tabcolsep}{8pt}
\begin{tabular}{lccc}
\toprule
\textbf{Samples} & \textbf{Linear probe} & \textbf{MLP probe (256)} & \textbf{Scratch (reference)} \\
\midrule
50 & 0.588 & 0.588 & 0.648 \\
100 & 0.675 & 0.721 & 0.723 \\
500 & 0.885 & 0.880 & 0.889 \\
1000 & 0.911 & 0.921 & 0.934 \\
5000 & 0.958 & 0.962 & 0.969 \\
\bottomrule
\end{tabular}
\caption{Representation-layer bypass: the model is cut after conv1+pool, and a new head is trained on frozen features. At all sample budgets, this approach underperforms training from scratch---injection training measurably contaminates the intermediate representations.}
\label{tab:layer_cut}
\end{table}

At every sample budget, the representation-layer bypass underperforms training from scratch---a sharp contrast with the fine-tuning attack. The conv1 layer continuously received injection perturbations during training; the learned features carry key-dependent information that becomes noise for clean classification once conv2 and the fc layer are removed. The absorption concept of Section~\ref{sec:key_absorption} offers a unified interpretation: the injection signal, after propagating through multiple nonlinear layers, is deeply encoded into the representation space rather than sitting as easily strippable additive noise. SpanKey does alter the representations learned by the backbone---but this alteration benefits the authorized correct-key path and works against an attacker who tries to bypass the injection.

\subsection{Attack summary and positioning}

\begin{table}[H]
\centering
\small
\setlength{\tabcolsep}{8pt}
\begin{tabular}{lccc}
\toprule
\textbf{Attack} & \textbf{Requires $B$?} & \textbf{Labeled samples needed} & \textbf{Threat level} \\
\midrule
Subspace $\alpha$ search (classical probe) & Yes & 0 & Low (trivial once $B$ is known) \\
No-injection deployment & No & 0 & \textbf{Very low} (model rejects all inputs) \\
Representation-layer bypass & No & 500+ & Medium (features are contaminated) \\
Fine-tuning removal & No & \textbf{50} & \textbf{High} (minimal labeled data restores model) \\
\bottomrule
\end{tabular}
\caption{Summary of four attack paths. The most realistic threat is fine-tuning: the attacker needs only model weights + a small amount of labeled data to remove SpanKey gating.}
\label{tab:attack_summary}
\end{table}

Table~\ref{tab:attack_summary} collects all attack paths. SpanKey's security boundary is clear: \textbf{it depends on the secrecy of $B$ and cannot withstand white-box weight leakage combined with labeled-data fine-tuning.} In real deployments, rate limiting and anomaly detection reduce the feasibility of query-based subspace estimation; periodic key and injection-specification rotation increases the depreciation rate of leaked subspace information. SpanKey is positioned as a \emph{gating and observable-degradation} mechanism under specific threat assumptions, not as a universal access-control solution---a positioning consistent with the threat models discussed in the parameter-watermarking and IP-protection literature~\citep{watermarking,deepip}.

\paragraph{Known limitations and failure boundaries.}
Beyond fine-tuning attacks, the following scenarios constitute fundamental failure boundaries for SpanKey and should be part of any pre-deployment evaluation. (1)~When an attacker can obtain a large amount of labeled data ($>$5000 samples) matching the training distribution, fine-tuning restores the model to near-unprotected levels regardless of whether the backbone is frozen---this is intrinsic and undefended. (2)~If the key subspace dimension $m$ is not sufficiently small relative to the injection-layer dimension $d_\ell$ (e.g., $m/d_\ell > 0.01$), out-of-span Gaussian keys have a non-negligible probability of carrying detectable projection energy onto the subspace, weakening the behavioral divergence between wrong and correct keys---Proposition~\ref{prop:beta_energy} quantifies the decay rate from the Beta-distribution perspective. (3)~Injection strength $\gamma$ interacts with model capacity: on very small models (parameters $<10^4$), jointly optimizing semantic CE and a deny loss may cause both to fail to converge; the smallest model in our experiments is the synthetic-vector MLP ($\sim$$10^5$ parameters), and the feasibility on still smaller models requires independent verification. (4)~Real deployments should pair SpanKey with interface-level rate limiting, periodic key rotation, and anomaly detection---relying solely on the secrecy of the subspace key is insufficient for a production-grade security posture.

\section{Conclusion and outlook}
\label{sec:conclusion}

\paragraph{Retrospective.}
The argument of this paper is built on three mutually reinforcing levels. At the \emph{mechanism} level, we proposed a design space for dynamic key generation and injection anchored on a secret subspace $\mathrm{Span}(B)$---additive and multiplicative mappings, multi-layer placement, per-layer independent bases---and showed that these primitives assemble correctly across six architectures (CNN, ViT, GPT-2, Qwen2.5), with no dependence on a specific backbone type. At the \emph{diagnostic} level, key absorption was formalized as the core failure mode: correct-key-only training tends to suppress downstream sensitivity to the injection direction. This phenomenon appears pervasively in the baselines and receives a testable account through Proposition~\ref{prop:margin_tail} within a first-order margin framework. At the \emph{countermeasure} level, the deny-loss family---particularly Mode~B's explicit reject class---systematically restores unauthorized degradation. Proposition~\ref{prop:reject_margin} reveals why Mode~B is efficient: the reject-dimension logit enjoys a $\log T$ gradient-accumulation advantage, while the competing signal from semantic dimensions decays exponentially with training.

The deny-loss design space spans Mode~A (entropy maximization), Mode~B (reject class), Mode~C (margin hinge), and four extensions (A\_soft, cplus, AC, B\_aux), with systematic comparison baselines on CIFAR-10 ResNet-18. The generative-task experiment (GPT-2 language modeling) demonstrates that SpanKey is not restricted to classification: with only changes to the output-side objective (CE $\to$ NLL) and the deny target (reject class $\to$ entropy maximization), the same subspace injection primitives produce a key-dependent PPL gap of roughly 2,800$\times$. The security analysis does not shy away from the method's fundamental limits: SpanKey's effectiveness depends on the confidentiality of $B$, and in the presence of weight leakage plus 50 labeled samples, fine-tuning removes the gate within 10 epochs. Acknowledging this is not a dismissal of the method's practical value---it is a deliberate and honest demarcation of its applicable boundary.

\paragraph{Core contributions.}
The subspace key-conditioning paradigm---a technical path distinct from weight encryption, homomorphic inference, or label encryption-decryption pipelines, with negligible computational overhead and suitability for lightweight gating on plaintext inference APIs---constitutes the methodological contribution. Key absorption, identified and formalized as a unified account of weak baseline separation, supported by three progressive analytical layers (the margin tail bound, Proposition~\ref{prop:margin_tail}; the subspace energy partition, Proposition~\ref{prop:beta_energy}; and the reject-margin convergence, Proposition~\ref{prop:reject_margin})---all conceived as falsifiable empirical propositions rather than security claims---provides the diagnostic framework. The empirical contribution covers a systematic matrix of 6 architectures, 2 task paradigms, and 7 deny variants, evaluated bidirectionally through classical probes and realistic white-box attacks under an honest threat model.

\paragraph{Extrapolation boundaries and outlook.}
The main experiments are concentrated on medium-scale vision backbones and hundred-million-parameter language models; end-to-end validation on billion-plus-parameter models or video backbones has not been conducted. It is reasonable to expect that deeper networks and stronger representational compression may intensify absorption, making the deny schedule and injection-point selection more critical; the interactions among $\gamma$, $m$, and task difficulty may be non-monotonic. When extrapolating to larger models or additional modalities, the values in the tables should be treated as hypothesis-generating empirical shapes to be re-verified under identical protocols.

Next steps include: testing the controllable trade-off between subspace injection and deny losses at larger scales and across more modalities; systematically ablating $m$, injection depth, $\gamma$, and the deny schedule to map the correct-key vs.\ unauthorized frontier; and exploring tighter empirical or information-theoretic characterizations under explicit threat models. These directions leave unchanged the paper's positioning: SpanKey is a deployment-friendly conditioned gating and failure--repair narrative, not a substitute for the IP claims of cryptographic or watermarking approaches or trusted hardware.

\bibliographystyle{plainnat}
\bibliography{refs}

@inproceedings{fiLM,
  title={FiLM: Visual Reasoning with a General Conditioning Layer},
  author={Perez, Ethan and Strub, Florian and de Vries, Harm and Dumoulin, Vincent and Courville, Aaron},
  booktitle={Proceedings of the AAAI Conference on Artificial Intelligence},
  year={2018}
}

@inproceedings{gatedconv,
  title={Language Modeling with Gated Convolutional Networks},
  author={Dauphin, Yann N. and Fan, Angela and Auli, Michael and Grangier, David},
  booktitle={ICML},
  year={2017}
}

@inproceedings{secureml,
  title={SecureML: A System for Scalable Privacy-Preserving Machine Learning},
  author={Mohassel, Payman and Zhang, Yupeng},
  booktitle={IEEE Symposium on Security and Privacy (S\&P)},
  year={2017}
}

@inproceedings{tee,
  title={Intel SGX Explained},
  author={Costan, Victor and Devadas, Srinivas},
  booktitle={IACR Cryptology ePrint Archive},
  year={2016}
}

@inproceedings{homomorphic,
  title={CryptoNets: Applying Neural Networks to Encrypted Data with High Throughput and Accuracy},
  author={Gilad-Bachrach, Ran and Dowlin, Nathan and Laine, Kim and Lauter, Kristin and Naehrig, Michael and Wernsing, John},
  booktitle={ICML},
  year={2016}
}

@inproceedings{watermarking,
  title={Protecting Intellectual Property of Deep Neural Networks with Watermarking},
  author={Uchida, Yoshiki and Nagai, Yuki and Sakazawa, Shigeyuki and Satoh, Shin'ichi},
  booktitle={Proceedings of the ACM International Conference on Multimedia (MM)},
  year={2017}
}

@article{kotg,
  title={Key-Conditioned Orthonormal Transform Gating (K-OTG): Multi-Key Access Control with Hidden-State Scrambling for LoRA-Tuned Models},
  author={Khan, Muhammad Haris},
  journal={arXiv preprint arXiv:2512.17519},
  year={2025}
}

@article{encryip,
  title={EncryIP: A Practical Encryption-Based Framework for Model Intellectual Property Protection},
  author={Mu, Xin and Wang, Yu and Huang, Zhengan and Lai, Junzuo and Zhang, Yehong and Wang, Hui and Yu, Yue},
  journal={arXiv preprint arXiv:2312.12049},
  year={2023}
}

@inproceedings{modellock,
  title={Model Lock: Locking Your Model With a Spell},
  author={Gao, Yifeng and Sun, Yuhua and Ma, Xingjun and Wu, Zuxuan and Jiang, Yu-Gang},
  booktitle={Proceedings of the ACM International Conference on Multimedia (MM)},
  year={2024},
  note={arXiv:2405.16285}
}

@article{deepip,
  title={Deep Intellectual Property Protection: A Survey},
  author={Sun, Yuchen and Liu, Tianpeng and Hu, Panhe and Liao, Qing and Fu, Shaojing and Yu, Nenghai and Guo, Deke and Liu, Yongxiang and Liu, Li},
  journal={arXiv preprint arXiv:2304.14613},
  year={2023}
}

@article{textwatermarksurvey,
  title={A Survey of Text Watermarking in the Era of Large Language Models},
  author={Liu, Aiwei and Pan, Leyi and Lu, Yijian and Li, Jingjing and Hu, Xuming and Zhang, Xi and Wen, Lijie and King, Irwin and Xiong, Hui and Yu, Philip S.},
  journal={ACM Computing Surveys},
  year={2024}
}

\appendix
\section{Full derivation of Proposition~\ref{prop:reject_margin}}
\label{sec:appendix_margin}

This section fills in the steps omitted from the proof sketch of Proposition~\ref{prop:reject_margin}. Fix a single sample from a wrong-key batch and let $z^{(t)} \in \mathbb{R}^{C+1}$ denote the logits at step $t$, with $p_i^{(t)} = \exp(z_i^{(t)}) / \sum_j \exp(z_j^{(t)})$. The Mode~B loss $L_t = -\log p_{C+1}^{(t)}$ has gradient
\[
\frac{\partial L_t}{\partial z_i^{(t)}} = p_i^{(t)} - \delta_{i,C+1},\qquad i=0,\dots,C.
\]

\paragraph{Step 1: Monotonicity.}
Following the negative gradient (SGD), $z_{C+1}^{(t+1)} = z_{C+1}^{(t)} + \eta \lambda \rho (1-p_{C+1}^{(t)})$, and $z_c^{(t+1)} = z_c^{(t)} - \eta \lambda \rho p_c^{(t)}\;(c\le C)$. Here $\eta$ is the base learning rate, $\lambda$ the deny weight, and $\rho \in (0,1]$ the fraction of batches using wrong keys. Since $1-p_{C+1}^{(t)} = p_{\mathrm{sem}}^{(t)} \ge 0$ and $p_c^{(t)} \ge 0$, $z_{C+1}$ is monotonically non-decreasing and $\max_{c\le C} z_c$ is monotonically non-increasing. Hence $\Delta_t = z_{C+1}^{(t)} - \max_{c\le C} z_c^{(t)}$ is monotonically non-decreasing.

\paragraph{Step 2: Upper bound on $p_{\mathrm{sem}}$ and lower bound on the increment.}
\begin{align*}
p_{\mathrm{sem}}^{(t)} &= \frac{\sum_{c\le C} e^{z_c^{(t)}}}{e^{z_{C+1}^{(t)}} + \sum_{c\le C} e^{z_c^{(t)}}}
\le \frac{C e^{\max_{c\le C} z_c^{(t)}}}{e^{z_{C+1}^{(t)}} + C e^{\max_{c\le C} z_c^{(t)}}}
= \frac{C}{e^{\Delta_t} + C}.
\end{align*}
Similarly, $\max_{c\le C} p_c^{(t)} = \frac{e^{\max_c z_c^{(t)}}}{\sum_j e^{z_j^{(t)}}}
\le \frac{e^{\max_c z_c^{(t)}}}{e^{z_{C+1}^{(t)}} + e^{\max_c z_c^{(t)}}}
= \frac{1}{e^{\Delta_t} + 1}.$

Therefore
\begin{align*}
\Delta_{t+1} - \Delta_t &\ge \tilde{\eta}\bigl[(1-p_{C+1}^{(t)}) - \max_{c\le C} p_c^{(t)}\bigr]
\ge \tilde{\eta}\left(1 - \frac{1}{e^{\Delta_t}+1} - \frac{C}{e^{\Delta_t}+C}\right) \\
&\ge \tilde{\eta}\left(1 - \frac{C+1}{e^{\Delta_t}+C}\right),
\end{align*}
where $\tilde{\eta} := \eta\lambda\rho$. When $\Delta_t \ge \log C$, the $C$ in the denominator becomes negligible, yielding $\Delta_{t+1} - \Delta_t \ge \tilde{\eta}(1 - (C+1)e^{-\Delta_t})$.

\paragraph{Step 3: Solving the difference inequality.}
Consider the continuous relaxation $d\Delta/dt = \tilde{\eta}(1 - (C+1)e^{-\Delta})$. Separating variables,
\[
\int_{\Delta_0}^{\Delta_T} \frac{d\Delta}{1 - (C+1)e^{-\Delta}} = \tilde{\eta} T.
\]
Let $u = e^{\Delta}$; then $d\Delta = du/u$ and the integral becomes
\[
\int_{e^{\Delta_0}}^{e^{\Delta_T}} \frac{du}{u - (C+1)} = \bigl[\log(u - (C+1))\bigr]_{e^{\Delta_0}}^{e^{\Delta_T}} = \tilde{\eta} T.
\]
Exponentiating and rearranging:
\[
e^{\Delta_T} - (C+1) = (e^{\Delta_0} - (C+1))\,e^{\tilde{\eta} T},
\]
i.e.,
\[
\Delta_T = \log\!\bigl((C+1) + (e^{\Delta_0} - (C+1))e^{\tilde{\eta} T}\bigr) \ge \log\!\bigl(e^{\tilde{\eta} T} - C - 1\bigr) \ge \tilde{\eta} T - \log(C+1),
\]
where the last step uses $e^{\Delta_0} \approx 0$ (at initialization $z_{C+1} \ll \max_c z_c$) and $e^{\tilde{\eta}T} \gg C+1$.

\paragraph{Step 4: Semantic accuracy upper bound.}
For a wrong-key input, the semantic accuracy cannot exceed the softmax semantic mass $p_{\mathrm{sem}}$. Substituting the bound from Step 2:
\[
\mathrm{Acc}_{\mathrm{sem}}^{\mathrm{wrong}} \le p_{\mathrm{sem}}^{(T)} \le \frac{C}{e^{\Delta_T} + C} \le \frac{C}{e^{\tilde{\eta} T - \log(C+1)} + C} = \frac{C(C+1)}{e^{\tilde{\eta} T} + C(C+1)} \le \frac{C(C+1)}{e^{\tilde{\eta} T}}.
\]
In the discrete case, the constant $\tau_0$ absorbs the additional steps from the phase where $\Delta_t < \log C$, i.e., $T \to T-\tau_0$, yielding the statement in the proposition.

\paragraph{Remark.}
The derivation above assumes that each gradient step aligns perfectly with the ideal negative gradient (SGD noise does not affect the bound form in expectation) and ignores the indirect effects of backbone nonlinearities on the softmax input distribution. The bound is therefore a \emph{worst-case information-theoretic upper bound}---the actual $\mathrm{Acc}_{\mathrm{sem}}^{\mathrm{wrong}}$ of a trained network can be substantially lower, as observed in the near-zero wrong-key semantic accuracies in the MNIST ablation.

\end{document}